\pdfoutput=1
\documentclass{llncs}
\usepackage{fainekos-macros}
\usepackage{algorithm}
\usepackage{algpseudocode}
\usepackage{graphicx}
\usepackage{subfig}
\usepackage{setspace}
\usepackage{ifthen}
\usepackage{amssymb}
\newtheorem{assumption}{Assumption}
\newboolean{Present}
\setboolean{Present}{true}
\newcommand*{\QEDA}{\hfill\ensuremath{\square}}%

\newboolean{TechReport}
\setboolean{TechReport}{true}

\begin{document}

	\title{Local Descent for Temporal Logic Falsification of Cyber-Physical Systems (Extended Technical Report)}%
%
\author{Shakiba Yaghoubi, \and Georgios Fainekos}
%
%
%
\institute{School of Computing, Informatics, and Decision Systems Engineering\\
Arizona State University, Tempe, AZ, USA\\
Email: \{syaghoub, fainekos\}@asu.edu }

\maketitle              

\begin{abstract}
One way to analyze Cyber-Physical Systems is by modeling them as hybrid automata. 
Since reachability analysis for hybrid nonlinear automata is a very challenging and computationally expensive problem, in practice, engineers try to solve the requirements falsification problem. 
In one method, the falsification problem is solved by minimizing a robustness metric induced by the requirements. 
This optimization problem is usually a non-convex non-smooth problem that requires heuristic and analytical guidance to be solved.
In this paper, functional gradient descent for hybrid systems is utilized for locally decreasing the robustness metric. 
The local descent method is combined with Simulated Annealing as a global optimization method to search for unsafe behaviors.

\end{abstract}

\section{Introduction}
\label{sec:intro}

\ifthenelse{\boolean{TechReport}}
{
In the last three decades, we have come to expect that we will be transported safely, reliably, and efficiently.
Technologically, we have reached this point by increasingly adding sensors and embedded computers in ground vehicles, airplanes, and locomotives.
However, as the software complexity increases so does the number of catastrophic software bugs.
Therefore, Model-Based Development (MBD) and auto-coding technologies are currently used as the preferred development method for reducing errors \cite{FerrariGMFT10fmics}.
Another benefit of MBD is that the system efficiency can be analyzed and optimized even before any prototypes are built.

Even though MBD can reduce coding errors, it may not necessarily reduce system design errors with respect to functional requirements.
This is particularly pronounced in Cyber-Physical Systems (CPS) where computer software interacts with and controls the physical environment. 
In order to analyze the safety of such systems, a variety of software tools have been developed for various classes of systems when the requirements concern reachable states in the system \cite{FrehseCAV11,ChenAS13cav}.
Nevertheless, when the system is complex in both the software and the physical dynamics, and the requirements have spatiotemporal constraints, e.g., as expressed in Metric \cite{Koymans90} or Signal \cite{MalerNickovic04} Temporal Logic (TL), then current reachability analysis methods cannot provide an answer.

 In order to address the need of providing real-time analysis of the behavior of such systems, a variety of search-based falsification methods has been developed (for a survey see \cite{KapinskiEtAl2016csm}).
}
{
In order to address the need for providing safety and real-time analysis for Cyber-Physical Systems (CPS), a variety of search-based falsification methods has been developed (for a survey see \cite{KapinskiEtAl2016csm}).
}
In search based falsification methods, the working assumption is that there is a design error in the system, and the goal of the falsifier is to search and detect system behaviors that invalidate (falsify) the system requirements.
\ifthenelse{\boolean{TechReport}}
{
Among search-based methods, multiple shooting optimization techniques \cite{zutshi2014multiple,zutshi2013trajectory} have shown great promise with a modest preprocessing stage, but still they cannot handle temporal logic requirements. 
Hence, tree search methods \cite{PlakuKV09tacas} and single shooting TL robustness guided approaches \cite{AbbasFSIG13tecs} remain the state of the art in TL falsification.
More recently, in \cite{DreossiDDKJD15nfm}, it was demonstrated that combining tree search based methods with TL robustness can improve the falsification detection rate in certain problem instances.
In brief, different falsification methods are still needed for different problem instances.
}
{
Typically, such requirements are formally expressed in Metric (MTL) \cite{Koymans90} or Signal (STL) \cite{MalerNickovic04} Temporal Logic (TL).
}

In this paper, we continue the progress on improving single shooting falsification methods for TL specifications \cite{AbbasFSIG13tecs}.
\ifthenelse{\boolean{TechReport}}
{
This class of methods is guided by evaluating how robustly a system trajectory satisfies a TL specification  \cite{FainekosP06fates,FainekosP09tcs}.
}
{
This class of methods is guided by evaluating how robustly a system trajectory satisfies a TL specification \cite{FainekosP09tcs}.
}
Positive values mean that the system trajectory satisfies the specification, while non positive values mean that the specification has been falsified by the system trajectory.
Single shooting falsification methods sample one or multiple system trajectories for the whole duration of the test time, they evaluate the TL robustness of each trajectory, and, then, they decide where to sample next in the search space.
Ideally, at each iteration, the proposed new samples will produce trajectories with TL robustness less than the previously sampled trajectories.
However, in general, this cannot be guaranteed unless some information is available about the structure of the system.
In \cite{abbas2014functional}, it was shown that given a trajectory of a non-autonomous smooth non-linear dynamical system and a TL specification, it is possible to compute a direction in the search space along which the system will produce trajectories with reduced TL robustness.
This direction is referred to as {\it descent direction} for TL robustness.

Our main contribution in this paper is that we extend the results of \cite{abbas2014functional} to computing local descent directions for falsification of TL specifications for hybrid systems. The extension is nontrivial since as discussed later in the paper, the sensitivity analysis is challenging in the case of hybrid systems.
\ifthenelse{\boolean{TechReport}}
{
In particular, we focus on hybrid automata \cite{tabuada2009verification,Alur15book} with non-linear dynamics in each mode and external inputs (non-autonomous systems).
Hybrid automata \cite{tabuada2009verification,Alur15book} is a mathematical model which can capture a wide range of CPS.
We also present several examples of hybrid automata for which we can compute such descent directions for TL specifications.
}
{
In particular, we focus on hybrid automata \cite{Alur15book} with non-linear dynamics in each mode and external inputs (non-autonomous systems).
Hybrid automata is a mathematical model which can capture a wide range of CPS.
}
We remark that the descent directions computed can only point toward local reduction of TL robustness.
Hence, we propose combining descent direction computations with a stochastic optimization engine in order to improve the overall system falsification rate.

We highlight that the contributions of this paper have some important implications.
First and foremost, it should be possible to derive results for approximating the descent direction for hybrid systems without requiring explicit knowledge of the system dynamics.
For example, in \cite{YaghoubiF2017acc}, we showed that this is possible for smooth non-linear dynamical systems by using a number of successive linearizations along the system trajectory.
The method was applied directly to Simulink models.
Second, the local descent computation method could be further improved by utilizing recent results on a smooth approximation of TL robustness \cite{PantAM17tech}.
Therefore, the results in this paper could eventually lead to testing methods which do not require explicit knowledge of the system dynamics, and could be applied directly to a very large class of models, e.g., Simulink models, without the need for model translations or symbolic model extraction.

\section{Problem Statement}
  In order to formalize the problem that we deal with in this paper, we will describe the system under test and also the system requirements in this section.
\subsection{System Description}
Hybrid automaton (HA) is a model that facilitates specification and verification of hybrid systems \cite{Alur15book}. 
A hybrid automaton is specified using a tuple $\Hc = (H, H_0 , U, Inv, \Ec, \Sigma)$,  
where $H= L \times X$ denotes the `hybrid' discrete and continuous state spaces of $\Hc$: $L \subset \Ne$ is the set of discrete states or locations that the system switches through (each location attributes different continuous dynamics to the system), and $X \subseteq \Re^n$ is the continuous state space of the system, $H_0 =L_0\times X_0\subseteq H$ is the set of \emph{initial conditions}, $U$ is a bounded subset of $\Re^m$ that indicates the input signals ranges, $Inv: L \rightarrow 2^{X \times \Re^+}$ assigns an invariant set to each location, $\Ec$ is a set of tuples $(E,Gu,Re)$ that determine transitions between locations. Here, $E \subseteq L \times L$ is the set of control switches, 
$Gu : E \rightarrow 2^{X\times \Re^+}$ is the guard condition that enables a control switch (i.e, the system switches from $l_i$ to $l_j$ when $(x(t),t)\in X\times \Re^+$ satisfies $Gu((l_i,l_j))$) and, $Re : E \times X\rightarrow X $ is a reset map that given a transition $e \in E$ and a point $x$ for which $Gu(e)$ is satisfied, maps $x$ to a point in the state space $X$. Finally, $\Sigma$ defines the continuous dynamics in each location $l\in L$:
     \begin{align}
\Sigma(l): \;\; {\textstyle\dot{x}}=F_l(x,u(t),t), \; x\in X, \; \forall t:\;u(t) \in U
\end{align}
where $\dot{x}= \frac{dx}{dt}$, $x \in X$ is the system continuous state, and $u: [0,T] \rightarrow U$ is the input signal map which is chosen from the set of all possible input signals $U^{[0,T]}$ whose value at time $t$ is denoted as $u(t)$. Also, $\forall l\in L, \;F_l: X\times U \times \Re_+ \rightarrow \Re$ is a $C^1$ flow that represents the system dynamics at location $l$.
\ifthenelse{\boolean{TechReport}}
{
For more information about hybrid systems please refer to \cite{tabuada2009verification} and \cite{Alur15book}.
}
{}
 
A {\it hybrid trajectory} $\eta{(h_0,u(t),t)}$ starting from a point $h_0=(l_0,x_0) \in H_0$ and under the input $u \in U^{[0,T]}$ is a function $\eta : H_0\times U\times \Re_+ \rightarrow H$ which points to a pair (control location, state vector) for each point in time: $\eta(h_0,u(t),t) = (l(h_0,u(t),t),s(h_0,u(t),t))$, where $l(h_0,u(t),t)$ is the location at time $t$, and $s(h_0,u(t),t)$ is the continuous state at time $t$.

We write the dynamical equations for the continuous system trajectory as: 
 \begin{align}
s(x_0,u(0)&,0)=x_0\nonumber\\
\frac{d s(x_0,u(t),t)}{d t}= F_l(s(x_0,u(t),t),u(t),t) \quad &\mbox{while} \; (s(x_0,u(t),t),t) \in Inv(l)\label{eq:2a} \\
s(x_0,u(t),t^+)= Re((l_i,l_j),s(x,u(t),t^-)) \;  &\mbox{if} \; {\scriptsize \left \{
\begin{array}{cc} 
(s(x_0,u(t),t^-),t) \in Gu((l_i,l_j)) \\
(s(x_0,u(t),t^+),t)\in Inv(l_j) 
\end{array} \right.}
\label{subeqn-2:hybdyn}
\end{align}
If the point $(s(x_0,u(t),t^+),t)$ lies outside $Inv(l_j)$, there is an error in the design. We assume that such errors do not exist in the system. The times in which the location $l$ and consequently the right-hand side of the equation (\ref{eq:2a}) changes, are called \emph{transition times}. In order to avoid unnecessary technicalities, in the above equations we use the notation of \cite{donze2007systematic} and denote transition times as $t^-$ and $t^+$, where $t^-$ is the time right before the transition and $t^+$ is the time right after that. 
However in more technical analysis of hybrid systems, one needs to consider the notion of hybrid time explained in \cite{goebel2006solutions} where a hybrid trajectory is parametrized not only by the physical time but also by the number of discrete jumps. 
When we consider the trajectory in a compact time interval $[0,T]$ and $\eta$ is not Zeno%
\footnote{$\eta$ is \emph{Zeno} if it does an infinite number of jumps in a finite amount of time. A hybrid system is Zeno if at least one of its trajectories is Zeno.},
the sequence of transition times is finite. 
\begin{assumption}
	\label{ass1}
	We assume our system is deterministic, it does not exhibit Zeno behaviors and given $(h_0,u)$ there is a unique solution  $\eta(h_0,u(t),t)$ to the system.
\end{assumption}
\begin{remark}
The input signal map $u$, should be represented using a combination of finitely many basis functions. In this paper we use piecewise constant signals. 
\end{remark}

\subsection{System Requirements}
Temporal logic formulas formally capture requirements concerning the system behavior.
They could be expressing the requirements over Boolean abstractions of the behavior using atomic propositions as in MTL \cite{Koymans90}, or directly through predicate expressions over the signals as in STL \cite{MalerNickovic04}.
Since the differences are only syntactic in nature (see \cite{DokhanchiHF15memocode}), in the following, we will just be using the term Temporal Logics (TL) to refer to either logic.
   
TL formulas are formal logical statements that indicate how a system should behave and are built by combining \emph{atomic propositions} (AP) or predicates using logical and temporal operators. 
The logical operators typically consist of {\it conjunction} ($\wedge$), {\it disjunction} ($\vee$), {\it negation} $(\neg)$, and {\it implication} ($\rightarrow $), while temporal operators include {\it eventually} $(\Diamond_\Ic)$, {\it always} $(\Box_\Ic)$ and {\it until} $(\Un_\Ic)$ where the index $\Ic$ indicates a time interval. 
\ifthenelse{\boolean{TechReport}}
{
For example, the specification: ``The absolute value of the trajectory $s$, should never go beyond $\alpha$'' can be captured using the TL formula $\Box(|s|\leq\alpha )$, or the timed specification ``The value of the signal $s$ should reach the bound $(s_{ref} \pm 5\%) $ within $\delta$ seconds and stay there afterwards'' can be formulated as $\Diamond_{[0,\delta]}(\Box(|(s-s_{ref})/s_{ref}|<5\%))$.
}
{
For example, the specification ``The value of the trajectory $s$ should reach the bound $(s_{ref} \pm 5\%) $ within $\delta$ seconds and stay there afterwards'' can be formulated as $\Diamond_{[0,\delta]}(\Box(|(s-s_{ref})/s_{ref}|<5\%))$.
}


The robustness of a trajectory $\eta(x_0,u,t)$ with respect to a TL formula is a function of that trajectory which shows how well it satisfies the specification (see \cite{FainekosP09tcs} for details on how the robustness is defined and calculated).  
The function creates a positive value when the requirement is satisfied and a negative value otherwise.
Its magnitude quantifies how far the specification is from being satisfied for non-positive values, or falsified for non-negative values.
Software tools such as  \staliro \cite{annpureddy2011s} compute the robustness value of a TL formula given a trajectory $\eta(x_0,u,t)$.
In order to detect unsafe system behaviors, we should falsify the specification, which means we need to find trajectories with non-positive robustness values. 
As a result, in a search based falsification, the effort is put on reducing the robustness value by searching in the parameter space. 

It can be easily shown that given a TL formula $\phi$ and a trajectory $\eta(h_0,u,t)$ of a hybrid automaton $\Hc$ that satisfies the specification, if Assumption \ref{ass1} holds, then there exists a \emph{critical time} $t^*\in[0,T]$ and a \emph{critical atomic proposition} (or {\it critical predicate}) $p^*$ with respect to which the robustness is evaluated \cite{abbas2013computing}. 
For example, in practice, the tool \staliro \cite{annpureddy2011s} computes the critical time $t^*$ and atomic proposition $p^*$ along with the robustness value of the specification.
Reducing the distance of the trajectory $\eta(h_0,u,t)$ from the set defined by $p^*$ at the critical time instance $t^*$ will not increase the robustness value; and in most practical cases it will actually decrease it. 
As a consequence, the TL falsification problem can be locally converted into a safety problem, i.e, always avoid the unsafe set $\Uc$ defined by $p^*$. 
Hence, we need to compute a descent vector $(h_0',u')$ that will decrease the distance between $\eta(h_0',u',t ^*)$ and the unsafe set $\Uc$. 



\vspace{-5pt}
\subsection{Problem Formulation}
%

Let $H_\Uc \subseteq H$ denote the system unsafe set, if $\eta(h_0,u(t),t)$ enters $H_\Uc$ then system specification is falsified.
To avoid a digression into unnecessary technicalities, we will assume that, both the set of initial conditions and the unsafe set are each included in a single control location, i.e, 
$H_0 = \{l_0\}\times X_0$, and $H_\Uc = \{l_\Uc\}\times \Uc$, where $l_0, l_\Uc \in L$, and $X_0,\Uc \subseteq X$. 



\begin{definition}
Let $D_{H_\Uc} : H \mapsto \Re_+$ be the distance function to $H_\Uc$, defined by
\begin{equation} 
\label{eq:DHU}
D_{H_\Uc}((l,x)) = \left \{
\begin{array}{cc} 
d_\Uc(x) & \mbox{ if }  l = l_\Uc \\
+\infty & \mbox{ otherwise } 
\end{array} \right .
\end{equation}
where $d_\Uc(x) = \inf_{\unspt\in \Uc} ||x-\unspt||.$
\end{definition}

	Given a compact time interval $[0,T]$, $h_0\in H_0$, and the system input $u\in U^{[0,T]}$, we define the robustness of the system trajectory $\eta(h_0,u(t),t)$ as
	\begin{equation}
	\label{eq:rob}
	f(h_0,u) \triangleq \min_{0 \le t \le T} D_{H_\Uc}(\eta(h_0,u(t),t))
	\end{equation} 
	and the respective critical time as
$t^* = \argmin_{t\in[0,T]} D_{H_\Uc}(\eta(h_0,u(t),t))$.
Since all trajectories start at $l=l_0$, we will write $f(h_0,u)$ as $f(w)$ where $w=(x_0,u)$. 
Trajectories of minimal robustness indicate potentially unsafe behaviors, and if we can reduce the robustness value to zero, we have a falsifying trajectory. As a result robustness value should be minimized with respect to $w$. Our problem can be formulated generally as follows:
\begin{equation}
	\label{eq:general}
\mbox{minimize} \;\: f(w) \;\: \mbox{such that} \;\: w \in X_0\times U^{[0,T]}
\end{equation}

\begin{figure}[t]
\centering
\includegraphics[width=7cm]{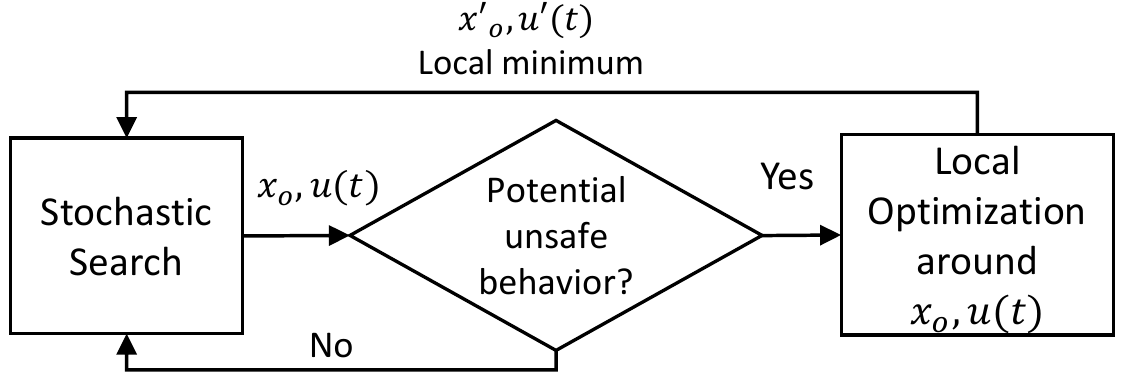}
\vspace{-5pt}
\caption{2-stage falsification: The stochastic search will search for the global optimizer while the local search improve the search speed.}
\vspace{-10pt}
\label{fig:2stage}
\end{figure}

Finding falsifying trajectories can be done in 2 stages. In the first stage, a higher level stochastic sampler determines a hybrid trajectory -a sequence of locations and state vectors- that exhibits system's potential bad behavior, and in the second stage, out of all the neighboring trajectories that follow the same sequence of locations, we find the trajectory of minimal robustness (see Fig (\ref{fig:2stage})). This can be done using local minimization. In this paper, we focus on solving the problem in this stage: we will find the trajectory of minimum robustness in the neighboring of a previously created trajectory in the first stage.

 Before we address our special problem of interest we should impose further assumptions on our system stated below:
\begin{enumerate}
\item The system is observable, i.e. we have access to all the system states, or we have a state estimator which is able to estimate them.

\item In the local search stage, we always are able to find a neighboring tube around each trajectory such that none of the trajectories inside that tube hit the guard tangentially. This ensures that trajectories of the system $\Hc$ starting close enough to $x_0$ and under neighboring inputs of $u$ undergo similar transitions/switches. In hybrid systems analysis, this property is called trajectory robustness (not to be confused with trajectory robustness in this paper) and is guaranteed if we can find an auto-bisimulation function of a trajectory and the trajectories starting from its neighboring initial conditions and under neighboring inputs \cite{winn2015safety}.
\item The system is deterministic and the transitions are taken as soon as possible. In order to have a deterministic system, if two transitions happen from the same location, their $Guards$ should be mutually exclusive.

\item $Guards$ are of the form $g(x,t)=0$ and $Reset$ maps are functions of the form $x'=h(x)$, where $g$ and $h$ are $C^1$ functions.  For all the states that satisfy a $Guard$ condition the corresponding $Reset$ map should satisfy $\frac{\partial h}{\partial x}\big\rvert _x\neq 0$.

\item The trajectory $\eta(h_0,u(t),t)$ returned by the first stage, from which we descend, enters the location of the unsafe set. 
\end{enumerate}

The last assumption is made so that our problem be well-defined (note that the objective function \eqref{eq:rob} will have finite value only if trajectory enters unsafe location). The task of finding such an initial condition $h_0$ is delegated to the higher-level stochastic search algorithm within which our method is integrated (Fig. \ref{fig:2stage}). If finding such a trajectory for the higher-level stochastic algorithm is hard, we can still improve our trajectories locally by descending toward the guards. This will be discussed more in the next section. 


The problem is addressed in the following:

 \begin{problem}
	Given a hybrid automaton $\Hc$, a compact time interval $[0,T]$, a set of initial conditions $H_0 \subseteq H$, a set of inputs $U^{[0,T]}$, a point $h_0 = (l_0,x_0) \in H_0$ and an input $u \in U^{[0,T]}$ such that the system trajectory satisfies $0 < f(w) < +\infty$, find a vector $dw=(dx_0,du)\in X\times U^{[0,T]}$ that satisfies the following property: 
 	
 	$\exists \Delta_1,\Delta_2 \in \Re^+$ such that $\forall \delta_1 \in (0,\Delta_1),\delta_2 \in (0,\Delta_2)$, $h_0' = (l_0,x_0+\delta_1 dx_0) \in H_0$ and $u'=u+\delta_2 du\in U^{[0,T]} $, $\eta(h'_0,u'(t),t)$ undergoes the same transitions as $\eta(h_0,u(t),t)$, and also $f(w+\delta dw) \le f(w)$ where $\delta=min\{\delta_1, \delta_2\}$.
 	\label{pb:main}
 \end{problem}

 Finding such a descent direction can help improve the performance of stochastic algorithms \cite{AbbasFSIG13tecs} that intend to solve the general problem in Eq. (\ref{eq:general}).

Note that for the piecewise constant inputs $u$ that we are working with in this paper, $du$ is also a piecewise constant signal whose variables should be computed. Variables of $du$ show the desired changes in that of the input signal $u$.
\section{Finding a descent direction for the robustness}
\label{des}
In this section, given a trajectory $\eta{(h_0,u(t),t)}$, we find $dx_0$ and $du$ such that the trajectory $\eta{(h'_0,u'(t),t)}$, where $h'_0=(l_0, x_0+\delta dx_0),u'(t)=u(t)+\delta du(t)$, attains a smaller robustness value; i.e $f(w')=f(x'_0,u')<f(x_0,u)=f(w)$. The robustness function in Eq. (\ref{eq:rob}) is hard to deal with as it is non differentiable and non convex \cite{abbas2013computing}. To solve this issue we calculate the descent direction with respect to a convex, almost everywhere differentiable function, and show that decreasing the value of this function yields a decrease in the robustness function:
\begin{theorem}
	\label{thm:substitute}
	Let $x_0,x'_0\in X_0$, $u,u'\in U^{[0,T]}$, and assume that the critical time for the continuous part of the hybrid trajectory  $s  \triangleq s({x}_0,{u(t)},t)$, is ${t}^*$. Define

\begin{equation} 
\label{jcon}
J(x'_0,u') = \left \{
\begin{array}{cc} 
\lVert {s}(x'_0,u'(t^*),{t}^*) - z({x_0},{u(t^*)},{t}^*)\lVert & \mbox{ if }  l = l_\Uc \\
+\infty & \mbox{ otherwise } 
\end{array} \right .
\end{equation}
	where $l$ is the first argument of $\eta(h'_0,u'(t^*),{t}^*)$, and 
	\begin{align}
	z({x_0},{u(t)},t) = \argmin_{z\in\mathcal{U}} \lVert z -s({x}_0,{u(t)},t)\rVert.
	\end{align}
	If we find a trajectory $s' \triangleq {s}(x'_0,u'(t),t)$ such that $J(x'_0,u') < J({x}_0,{u})$, 
	then the robustness of the trajectory ${s'}$ is smaller than that of $s$, i.e:
	$f(x'_0,u') < f({x}_0,{u})$.
\end{theorem}
\begin{proof}
	By Eq. (\ref{eq:rob}) we have
	$\displaystyle f(x'_0,u') = \min_{0 \le t \le T} D_{H_\Uc}(\eta(h'_0,u'(t),t))\leq J(x'_0,u') < J({x}_0,{u}) = f({x}_0,{u})$.
    \QEDA
    \end{proof}
 Let $x'_0=x_0+dx$ and $u'=u+du$. Consider $J$ at the unsafe location and define:
\begin{align}
	J(x'_0,u') = G(s(x'_0,u'(t^*),{t}^*)),
\end{align}
where $G(x)=\lVert x - z({x_0},u(t^*),{t}^*)\lVert $. Notice that the definition of $G$ is based on a primary trajectory from which we want to descend. The total difference of a multi variable function shows the change in its value with respect to the changes in its independent variables while its partial differential is its derivative with respect to one variable, while others are kept constant. In the following, $dx$ and $du$ are calculated using the chain rule, such that $J(x'_0,u') - J({x}_0,{u})=J(x_0+dx,u+du) - J({x}_0,{u})=dJ({x}_0,{u})<0$:
\begin{equation} 
\label{dj}
dJ(x_0,u;dx,du) = \frac{\partial G}{\partial x}^Tds(x_0,u,t^*)
\end{equation}
where $\frac{\partial G}{\partial x}\triangleq\left.
\frac{\partial G}{\partial x}\right\vert_{s(x_{0},u(t^*),t^*)}
\in\mathbb{R}^{n\times 1}$ is the steepest direction that increases distance from the unsafe set, i.e, $-\frac{\partial G}{\partial x}$ is along the approach vector mentioned in  \cite{abbas2013computing} that shows the direction of the shortest distance between $s(x_0,u(t^*),t^*)$ and the unsafe set. Now observe that:
\begin{equation} 
\label{ds}
ds(x_0,u,t^*) = D_1 s(x_0,u,t^*)dx_0+D_2 s(x_0,u,t^*)du
\end{equation}
where $D_i$ denotes the partial differentiation with respect to the $i^{th}$ argument (for instance $D_1s=\frac{\partial s}{\partial x_0}$). Here, $D_1 s(x_0,u,t^*)$ and $D_2 s(x_0,u,t^*)$ are the sensitivity of the trajectory to the initial condition and input at time $t^*$, respectively. In the next section we show how to calculate sensitivity for a hybrid trajectory. Using Eq. (\ref{dj}) and (\ref{ds}), we choose: 
\begin{align}
\label{dxu}
dx_0=-c_1(\frac{\partial G}{\partial x}^T D_1 s(x_0,u,t^*))^T, \quad
du=-c_2(\frac{\partial G}{\partial x}^T D_2 s(x_0,u,t^*))^T
\end{align}
for some $c_1,c_2>0$. As a result, we have $\label{dj2}
dJ(x_0,u) = -c_1||\frac{\partial s}{\partial x_0}^T\frac{\partial G}{\partial x}||^2-c_2||\frac{\partial s}{\partial u}^T\frac{\partial G}{\partial x}||^2$ $
\leq 0$ and the equality holds if and only if $\frac{\partial s}{\partial x_0}{\big\rvert_{(x_{0},u,t^*)}}^T\frac{\partial G}{\partial x}=\frac{\partial s}{\partial u}{\big\rvert_{(x_{0},u,t^*)}}^T\frac{\partial G}{\partial x}=0$.

 All the above calculations are based on the assumption that the trajectory enters the unsafe location, but even if finding a trajectory that enters the unsafe location using stochastic higher level search is hard, we can still improve trajectories locally by descending toward the guard $Gu^*$ that takes the trajectory to the location with the shortest possible path to the unsafe set. This is shown in Fig. \ref{fig:guarddes}. For instance if the guard $Gu^*$ is activated when $g(x)=0$, we can easily use zero finding methods to find a set $M=\{x\;|\; g(x)=0\}$ and replace $\Uc$ in all the previous calculations with the set $M$. 


\section{Sensitivity Calculation for a Hybrid Trajectory }\label{sec4}
Extending sensitivity analysis to the hybrid case is not straightforward and even
in the case that there is no reset in transitions and the state stays continuous, a discontinuity
often appears in the sensitivity function that needs to be evaluated \cite{donze2007systematic}. In order to make the results comprehensive, in this section we analyze the sensitivity for trajectories of a Hybrid automaton. Without loss of generality, in order to focus on the complexity that happens under transitions, we consider a hybrid automaton with only two discrete locations ($|L|=2$) and one control switch, also we assume $l_0\neq l_{\Uc}$. There are 2 scenarios: 
\begin{enumerate}
	\item 	$(s(x,u(t),t),t)$ is either inside $Inv(l_0)$ or $Inv(l_{\Uc})$
	\item 	$(s(x,u(t),t),t)  \in Gu((l_0,l_{\Uc}))$
\end{enumerate}

Let us use $p_{x_0}$ and $p_{u}$ to denote the sensitivity of the trajectory to changes in $x_0$ and $u$ respectively, i.e, $p_{x_0}(t,t_0)=D_1{s(x_0,u,t)}$ and $p_{u}(t,t_0)=D_2{s(x_0,u,t)}$.
	It can be shown easily that in the first scenario, while $(s(x_0,u,t),t) \in Inv(l_i)$ and $i\in\{0,\Uc\}$:
	 \begin{subequations}
	 	\begin{align}
	 		\dot{p}_{x_0}(t,t_0) = D_1  F_{l_i}(&s(x_0,u,t),u(t),t).p_{x_0}(t,t_0), \\
	 	\dot{p}_{u}(t,t_0) = D_1  F_{l_i}(s(x_0,u(t),t),&u(t),t).p_{u}(t,t_0) +D_2  F_{l_i}(s(x_0,u(t),t),u(t),t),	
		 	\end{align}
	 \end{subequations}
 with the following initial and boundary conditions:
 \begin{subequations}\label{eqn:bound}
	 	\begin{align}
	 	{p}_{x_0}(t_0,t_0)&=I_{n\times n}, \; {p}_{u}(t_0,t_0)=0,\\
	 	p_{x_0}(\tau^+,t_0)     &= r_{x_0}, \; 
	 	p_{u}(\tau^+,t_0)    = r_{u}.
	 	\end{align}
	 \end{subequations}
	 where $\tau^+$ is the right hand side limit of the transition time $\tau$ that satisfies $(s(x_0,u(\tau),$ $\tau),\tau)\in Gu((l_0,l_{\Uc}))$. We will calculate $r_{x_0}$ and $r_u$ in the following subsection. Consider that even if there is no reset, this jump happens in the state triggered transitions since neighboring trajectories have different transition times and as a result they are under different dynamics during the time between their transition times (see Fig. \ref{fig:jumps}).
     \begin{figure}[t]
            \centering
            \begin{minipage}[b]{0.47\textwidth}
	    	\includegraphics[width=\linewidth]{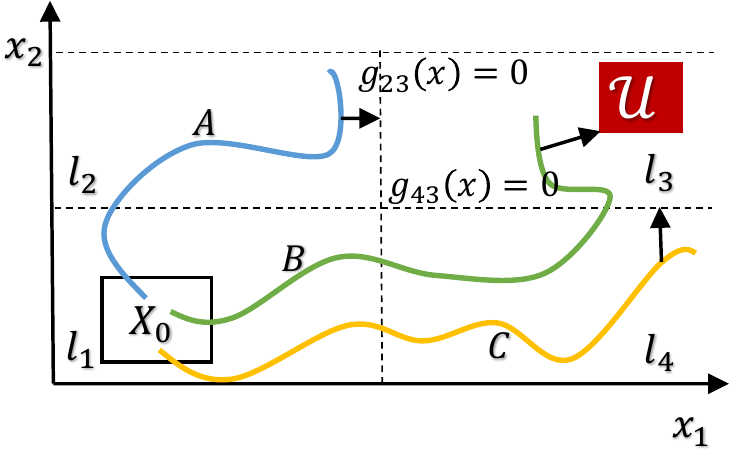}
\vspace{-15pt}
\caption{Trajectories B, A and C improve locally by descending toward the unsafe set, guard $g_{43}$ and guard $g_{23}$ respectively.}
\label{fig:guarddes}
            \end{minipage}
            \hfill
            \begin{minipage}[b]{0.47\textwidth}
            \vspace{-10pt}
           \includegraphics[width=0.9\linewidth]{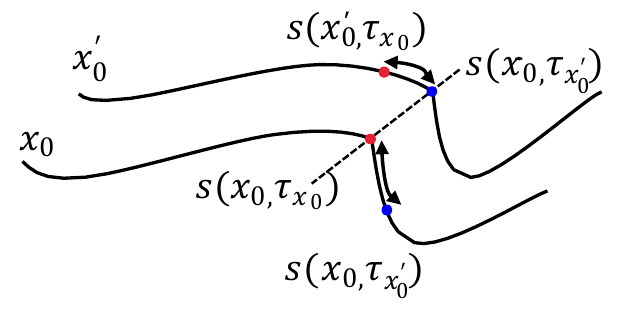}
          \vspace{-5pt}
          \caption{ Assuming $\tau_{x_0}<\tau_{x'_0}$, trajectories are under different dynamics for all the times $t\in[\tau_{x_0},\tau_{x'_0}]$, where $\tau_{x_0}$ and $\tau_{x'_0}$ are transition times for $s(x_0,.)$ and $s(x'_0,.)$ respectively.}
           \label{fig:jumps}
            \end{minipage}
\end{figure} 
     
\subsection{Sensitivity Jump Calculation}
  Assume that if $g(s(x_0,u(t),t),t)=0$ then $(s(x_0,u(t),t),t)\in Gu((l_0,l_{\Uc}))$. Let us denote the transition time by $\tau(x_0,u)$, which reminds us that this transition time differs for different trajectories; if the dependence was clear from context, we will write down $\tau$, for brevity. Assume that $Re(x,(l_1,l_2))=h(x)$, we have:
 \begin{align}
 s(x_0,u(\tau^+),\tau^+)=h(s(x_0,u(\tau^-),\tau^-))
 \end{align}
To calculate the value of $p_{x_0}$ at $\tau^+$ we take derivatives with respect to $x_0$ from the above equation. We have:
 \begin{align*}
\frac{ds(x_0,u,\tau^+)}{dx_0}=\frac{\partial h}{\partial x}&\frac{ds(x_0,u,\tau^-)}{dx_0} \Rightarrow
\end{align*} \vspace{-12pt} 
 {\small\begin{align*}
 {D_1 s(x_0,u,\tau^+)+D_3s(x_0,u,\tau^+)\frac{\partial \tau}{\partial x_0}=  \frac{\partial h}{\partial x}(D_1(s(x_0,u,\tau^-))+D_3s(x_0,u,\tau^-)\frac{\partial \tau}{\partial x_0})} 
 \end{align*}}\vspace{-12pt}
\begin{align}
\Rightarrow p_{x_0}(\tau^+,t_0)=D_1 s(x_0,u,\tau^+)=\frac{\partial h}{\partial x}p_{x_0}(\tau^-,t_0)+(&\frac{\partial h}{\partial x}f^--f^+)D_1\tau) \label{eq1}
 \end{align}
 where $\frac{\partial h}{\partial x}=\frac{\partial h}{\partial x}\big \rvert _{s(x_0,u(\tau^-),\tau^-)}$, and $f^-$ and $f^+$ are equal to $F_{l_0}(s(x_0,u(\tau^-),\tau^-)$, $u(\tau^-),\tau^-)$ and $F_{l_\Uc}(s(x_0,u(\tau^+),\tau^+),u(\tau^+)),\tau^+)$ respectively.
  To calculate $D_1\tau $, consider that $\tau$ satisfies $g(s(x_0,u,\tau),\tau(x_0,u))=0$, taking the derivatives with respect to $x_0$, we have:\vspace{-5pt}
  \begin{align}
 D_1g^T(D_1s(x_0,u,\tau)&+D_3s(x_0,u,\tau).D_1\tau )+
D_2g.D_1\tau =0\nonumber\\
\Rightarrow D_1\tau =\frac{\partial \tau}{\partial x_0}&=-\frac{D_1g^T.p_{x_0}(\tau^-,t_0)}{D_1g^T.f^-+D_2g} \label{tx}
  \end{align}

Using similar analysis we have:
  \begin{align}
  p_{u}(\tau^+,t_0)=\frac{\partial h}{\partial x}&p_{u}(\tau^-,t_0)+(\frac{\partial h}{\partial x}f^--f^+)D_2\tau ^T\\
  D_2\tau=&-\frac{D_1g^T.p_{u}(\tau^-,t_0)}{D_1g^T.f^-+D_2g}  \label{tu}
  \end{align}
  Using a hybrid automaton, sensitivity and system states can be calculated simultaneously (see Fig. \ref{fig:sensauto}). This will easily let us calculate the sensitivities by reseting their values at transition times. 
    \begin{figure}[t]
    	\centering
    	\includegraphics[width=8.5cm]{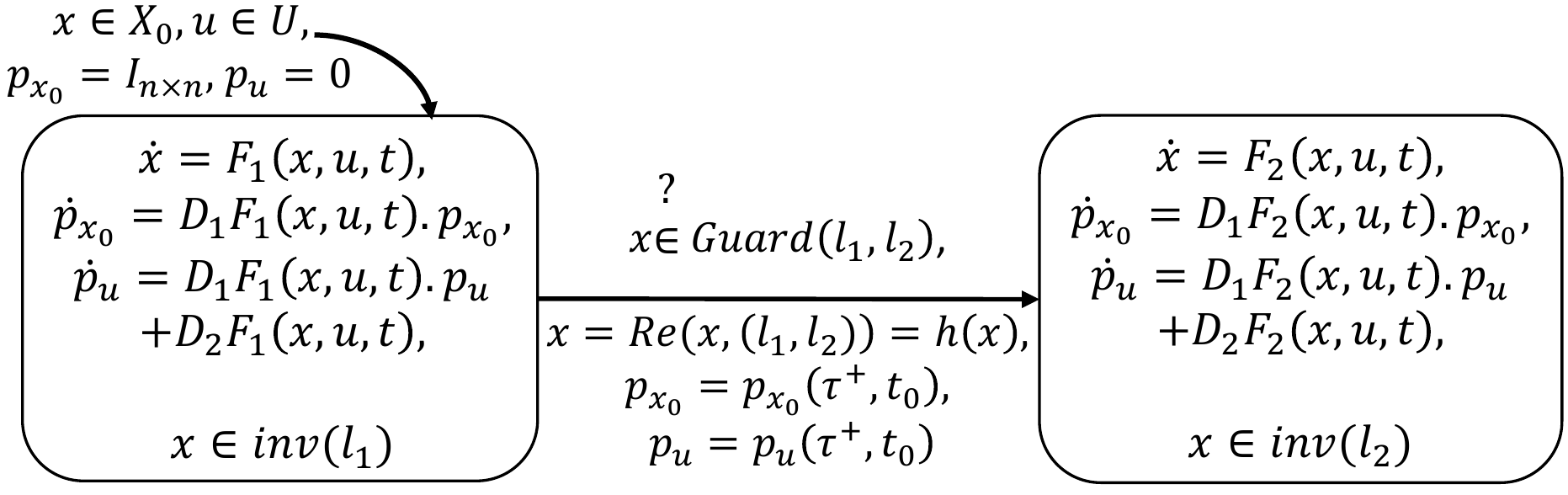}
    	\caption{HA of the system and trajectory sensitivity}
    	\label{fig:sensauto}
    \end{figure}
Note that using equations (\ref{eq1}) to (\ref{tu}), for a system with time triggered transitions ($g(x,t)=g'(t)$) whose reset map is identity ($h(x)=x$), there are no jumps in sensitivities, i.e, $p_{x_0}(\tau^+,t_0)=p_{x_0}(\tau^-,t_0)$ and $p_{u}(\tau^+,t_0)=p_{u}(\tau^-,t_0)$. These types of hybrid systems can be handled using our previous work in {\cite{YaghoubiF2017acc}} where we showed how to use system linearized matrices to approximately calculate the decent direction. However to have these kinds of gray box analysis for hybrid systems with state dependent transitions, we also need to have some information about the guards or be able to approximate them in order to model the jumps in the sensitivity. In the future we will work on the descent calculation using gray box models of the general hybrid systems. 

\ifthenelse{\boolean{Present}}{Algorithm (\ref{alg:GD}) describes the procedure to find gradient descent (GD) directions for hybrid systems. The function ``\text{\sc Simul\&Sens\&GetRob}'' calculates the sensitivity matrices $p_{x_0}$ and $p_u$ as well as the robustness value $r$, the critical time $t^*$ and the approach vector $n_s=\frac{\partial G}{\partial x}$ with respect to the specification $\varphi$. Given $t^*,n_s,p_{x_0},p_u$, ``GD'' calculates the gradient descent directions $dx,du$ using Eq. (\ref{dxu}). ``INBOX'' calculates the new initial condition and input while ensuring that they lie inside the desired sets $X_0$ and $U^{[0,T]}$.

\begin{algorithm}[t]
		\caption{ Robustness Gradient Descent algorithm}
		\label{alg:GD}
		\begin{algorithmic}[1]
			\Require Hybrid system model $\Hc$, initial condition and input $x_0$ and $u$, sets of possible initial conditions and input values $X_0$ and $U$, system specification $\varphi$, finial time $T$, step size $h$, maximum number of iterations we descend $k_1$, maximum number of iterations we decrease the step size $k_2$ and the multiplier of the step size $p<1$.
		\Ensure local optimal initial condition $x_0^*$, local optimal input $u^*$ and the related optimal robustness value $r^*$			\State $(x_0',u',r^*)$ $\gets$ $(x_0,u_0,\infty)$
            \For {$i= \; 1 \; to \; k_1$}
			\State $(r,t^*,n_s,p_{x_0},p_u)$ $\gets$ \text{\sc Simul\&Sens\&GetRob}$(x'_0,u',\Hc,T,\varphi)$
			\If {$r \leq r^*$}
			\State $(x_0,u)$ $\gets$ $(x_0^\prime,u^\prime)$, \; $(x_0^*,u^*,r^*)$ $\gets$ $(x_0^\prime,u^\prime,r)$
			\Else
            \State $h'$$\gets$$h$
			\For {$j= \; 1 \; to \; k_2$}
			\State $h'$ $\gets$ $h'.p$
			\State $(x_0^\prime,u^\prime)$ $\gets$  \text{\sc inbox}$(x_0,u,X_0,U_0,h',d_x, d_u)$
				\State $(r,t^*,n_s,p_{x_0},p_u)$ $\gets$ \text{\sc Simul\&Sens\&GetRob}$(x'_0,u',\Hc,T,\varphi)$
            \If {$r \leq r^*$}
			\State $(x_0^*,u^*,r^*)$ $\gets$ $(x_0^\prime,u^\prime,r)$
			\State $\mathbf{Break}$
			\EndIf
			\EndFor
			\EndIf
            \State $(dx,du)$ $\gets$ \text{\sc GD}$(t^*,n_s, p_{x_0}, p_u)$.
			\State $(x_0^\prime,u^\prime)$ $\gets$  \text{\sc inbox}$(x_0,u,X_0,U_0,h,d_x, d_u)$
			\EndFor
		\end{algorithmic}
    \end{algorithm}}{An algorithm to find the gradient descent (GD) directions for hybrid systems is mentioned in the technical version of the paper \cite{techversion}.} 
\setlength{\textfloatsep}{5pt}

\section{Experimental Results} 
\ifthenelse{\boolean{Present}}{
In order to show the utility of our method, we used the following three examples in which we deal with nonlinear hybrid systems. In all the experimens we used MATLAB 2015b on an Intel(R) Core(TM) i7-4790 CPU @3.6 GHZ with 16 GB memory processor with Windows Server 2012 R2 Standard OS.

	\begin{example}
The first example models the motion of a billiard ball. 
The ball is initially placed at $(x_0,y_0)$, and it is shot in direction $a$  with speed $v$ where $a$ is the throw angle with the $x$-axis. We assume there is no friction between the table and the ball. When the ball hits the sides parallel to $x$-axis (lines $y=0,2$), its velocity on the $y$-direction flips sign while the velocity in $x$-direction remains unchanged. Also, when it hits the side $x=4$, its velocity on the x-direction flips sign and the velocity on the $y$-direction remains unchanged. That is, the collisions of the ball with the table sides are perfect and no energy is lost.
The system can be modeled using a simple hybrid automaton, as shown in Fig (\ref{fig:pool}). 

Assume, we want to hit the ball such that it eventually falls in the hole centered in $(0.2,1.6)$ with radius $0.1$. The search is done over 3 Dimensions: We allow the ball to be initially placed at $[0,0.2]^2$ and be shot using angle $a \in [30^{\circ},45^{\circ}]$ with $v=1$.
	Starting from the initial condition $(0.1,0.1)$, and using $a=48.5^{\circ}$, the throwing process is refined using GD method. The trajectories are shown in Fig (\ref{fig:ex1}) where we refined light gray trajectories to the darker ones.
		\end{example}
    	
        \begin{example}
		Our second example is a hybrid model of glycemic control in diabetic patients taken from \cite{chen2012taylor} in which they used feedback control strategies by \cite{furler1985blood} and \cite{fisher1991semiclosed}. The variation of insulin glucose levels in diabetic patients is modeled using the following equations:
                    \vspace{-4pt}
        {\small	\begin{align}
			 \begin{bmatrix}
 	    \dot{G}        \\
 	    \dot{X}   \\
 	    \dot{I}
 	\end{bmatrix}=\begin{bmatrix}
 	    -p_1 G-X(G+G_B)+u_1(t) \\
 	     -p_2X+p_3I   \\
 	    -n(I+I_b)+\frac{u_2(t)}{v_I}
 	\end{bmatrix}
			\end{align}}
                        \vspace{-4pt}

          where the state $G$ is the level of glucose in the blood above the basal value $G_B=4.5$, $X$ is proportional to the insulin level that is effective in glucose level control, and $I$ is the insulin level above the value $I_b=15$. Typical parameter values for $p_2$, $V_I$, and $n$ are 0.025, 12 and 0.093, respectively, and parameters $p_1$ and $p_3$ are patient dependent. The functions $u_1(t)$ and $u_2(t)$ model the infusion of glucose and insulin into the bloodstream in order to control their levels, and their values are chosen based on the following equations:
             \vspace{-2pt}
 {\small \begin{align}
 		u_1(t)=\begin{cases}
		1+\frac{2G(t)}{9} \qquad G(t)<6\\
		\frac{50}{3}  \qquad G(t) \geq 6
		\end{cases}
		,\; u_2(t)=\begin{cases}
		\frac{t}{60} \qquad t \leq 30\\
		\frac{120-t}{180} \qquad 30 \leq t\leq 120\\
		0 \qquad t \geq 120
		\end{cases} \nonumber
		\end{align} }

\begin{figure}[tbp]
  \centering
  \begin{minipage}[b]{0.3\textwidth}
\includegraphics[width=4.5cm]{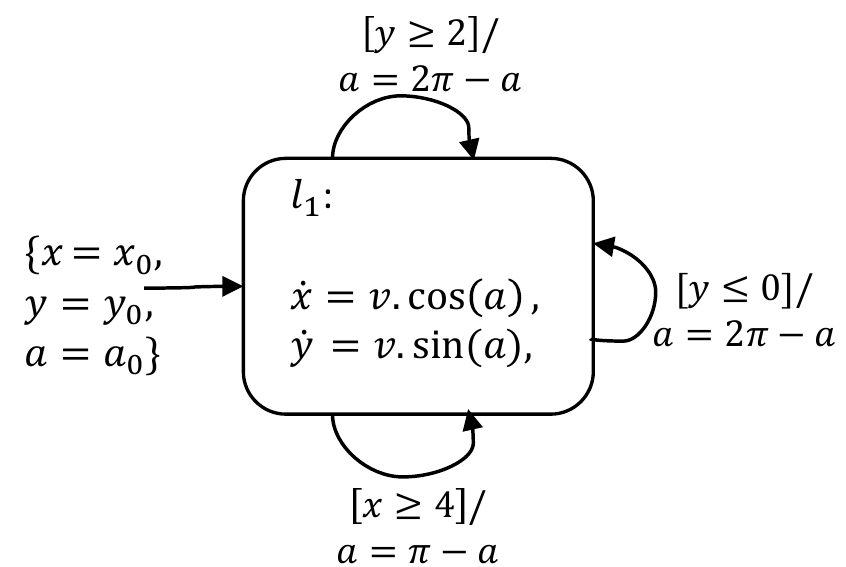}
\caption{Billiard ball hybrid automaton.}
\label{fig:pool}
\vspace{+15pt}
  \end{minipage}
  \hfill
  \begin{minipage}[b]{0.6\textwidth}
\includegraphics[width=\linewidth]{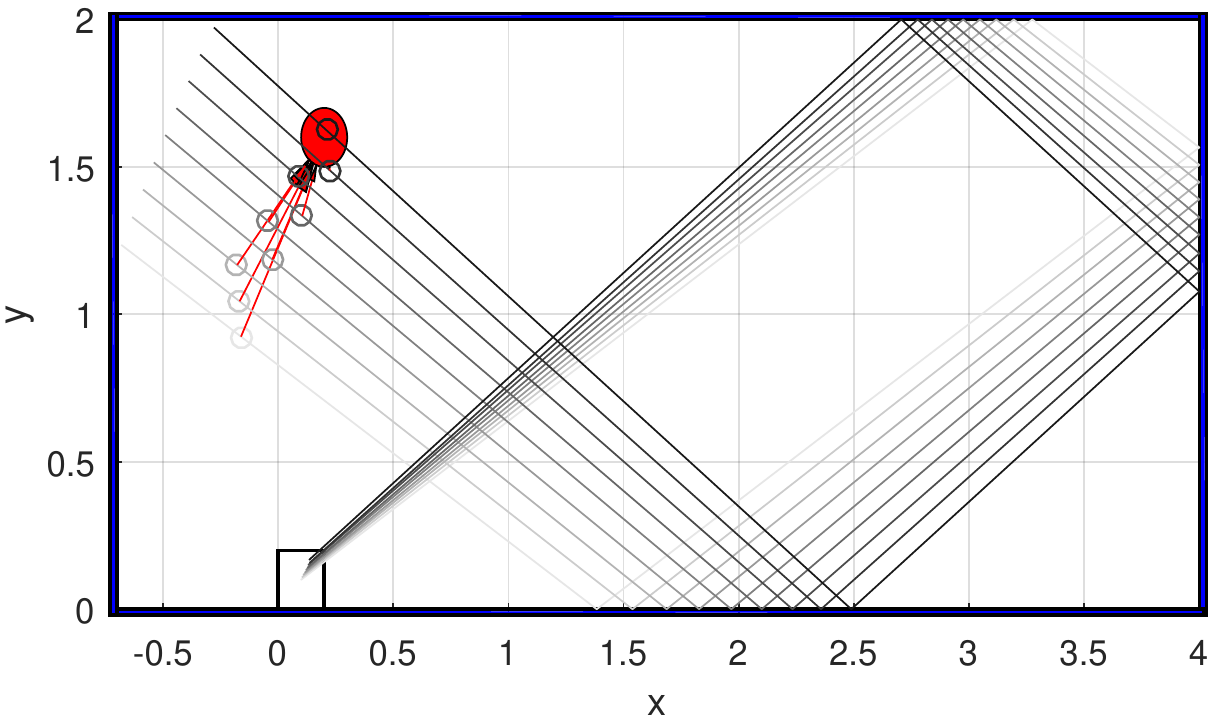} 
\vspace{-15pt}
\caption{Starting in the white box and using $a\in[30^{\circ},45^{\circ}]$, the ball should be placed in the red hole. The red arrows show the descent direction.}
\label{fig:ex1}
\vspace{-10pt}
  \end{minipage}
\end{figure}    
	Using the above control schemes for $u_1(t)$ and $u_2(t)$ yields a hybrid automaton with 6 locations/modes with 4 timed-based and 6 state-based guards. While the sensitivity-states $p_{x_0}$ and $p_u$ go under jumps during the state-based transitions, they remain unchanged in timed transitions. Note that based on Eq. (\ref{tx}) and (\ref{tu}), in timed transitions, $D_1g$ and as a result $\tau_{x_0}=\frac{\partial \tau}{\partial x_0}$ and $\tau_u=\frac{\partial \tau}{\partial u}$ are zero. In this example the search is over 5 dimensions: $[G,X,I]\in [6,9.5]\times[0.15,0.18]\times[-0.1,0.1]$ and $p_1\in[0,0.02]$ and $p_3 \in [10^{-5},10^{-4}]$. The system should satisfy:
		\begin{align}
		\varphi_1=&\Box_{[0,30]}G\in[-3,10] \wedge \Box_{[30,120]}G\in[-1.5,5.1] \wedge \Box_{[120,200]}G\in[2,5] \nonumber
		\end{align}
		Starting from $x_0=(6.5,0.17,0)$, and using $(p_1,p_3)=(0.01,1.3\times 10^{-5})$ with robustness 0.8287, the optimization process reduces the robustness to -0.0213 using $x_0=(6.5001,0.1506,-3.064\times10^{-6})$, and $(p_1,p_3)=(0.0097,1\times 10^{-4})$ which results in falsification of the requirement. Figure (\ref{fig:ex2}) shows the glucose trajectories $G(t)$ where the search is started using the light gray trajectories and refined to the darker ones. Note that because of the local search property of the method, for trajectories in this this local search, the effort is put on decreasing the distance to the critical unsafe set, which is the set $[5.1,\infty)$ at the critical time $t^*\in [30,120]$.
		\begin{figure}[t]
		 	\centering
		 	\includegraphics[width=7 cm]{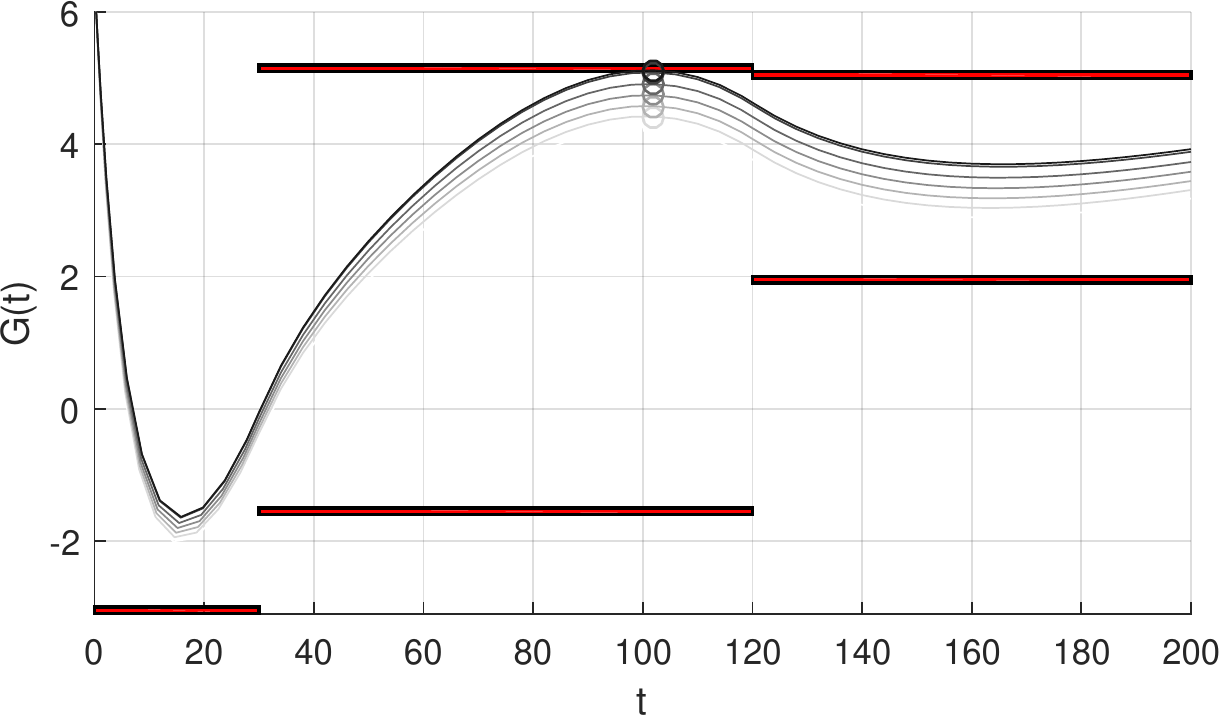}
            \vspace{-5pt}
		 	\caption{Falsification of the specification $\varphi_1$ for the glycemic control model}
		 	\label{fig:ex2}
		\end{figure}
	\end{example}}

		\ifthenelse{\boolean{Present}}{\begin{example} \label{ex3}
	Our last example is a rotating planar vehicle.}{In order to show the utility of our method, in \cite{techversion} we used three examples in which we deal with nonlinear hybrid systems. In all the experiments we used MATLAB 2015b on an Intel(R) Core(TM) i7-4790 CPU @3.6 GHZ with 16 GB memory processor with Windows Server 2012 R2 Standard OS. In the following we present one of the examples:
    \begin{example} \label{ex3}
} Consider the motion of a rigid object on a plane that uses a pair of off-centered thrusters as the control input. Since these thrusters are not aligned with the center of the mass, they will create both translational and rotational motions on the vehicle \cite{winn2015safety}. \ifthenelse{\boolean{Present}}{(see Fig. \ref{fig:thrust} for a better illustration). T}
The system is supposed to satisfy the requirement in Eq. (\ref{req}) which implies that the vehicle should avoid the unsafe sets $\Uc_1$ and $\Uc_2$ (shown in Fig. \ref{fig:car} with red boxes) and reaches the goal set $G$ (shown with a blue box) within the simulation time $T=10$. Here $(x_1,x_2)$ is the vehicle position.
    	    \begin{align}
	    \varphi_2=\Box_{[0,10]} \neg ((x_1,x_2)\in \Uc_1 \vee (x_1,x_2)\in \Uc_2) \wedge \Diamond_{[0,10]} (x_1,x_2)\in G  \label{req}
	    \end{align}        

The location-based dynamics of the vehicle are mentioned in Eq. (\ref{car}), where $j\in\{1,2,3\}$, $x_1,x_2$ are the positions along the $x$ and $y$ axis, $x_3$ is the angle with the $x$-axis and $x_4,x_5$ and $x_6$ are their derivatives. The hybrid model consists of 3 locations, where $inv(l=1)=\{x|x_1 <4\}$, $inv(l=2)=\{x|4 \leq x_1 \leq 8\}$, and $inv(l=3)=\{x|x_1>8\}$. The guards are shown using dashed lines in Fig. \ref{fig:car}. 
The unsafe sets have attractive non-centered forces in their corresponding locations. 
In particular, $\Uc_1$ is located in location 2 and $\Uc_2$ is located in location 3. 
At location 1, $s_1(l=1)=s_2(l=1)=0$, at location 2, $s_1(l=2)=-1$ and $s_2(l=2)=0$, and at location 3, $s_1(l=3)=0$ and $s_2(l=3)=-2$. $(\alpha_1,\beta_1)$ and $(\alpha_2,\beta_2)$ are the centers of $\Uc_1$ and $\Uc_2$, respectively.	  
	     \begin{align}
	    	 \begin{bmatrix}
	    	\dot{x}_j \\
	    	\dot{x}_4 \\
	    	\dot{x}_5 \\
	    	\dot{x}_6 \\
	    	\end{bmatrix}=
	    	\begin{bmatrix}
	    	x_{j+3} \\
	    {\scriptstyle 0.1x_4+\Sigma_{i=1,2}s_i(l)(x_1-\alpha_i)+F_1cos(x_5)-F_2sin(x_5)}\\
	    	 {\scriptstyle 0.1x_4+\Sigma_{i=1,2}s_i(l)(x_2-\beta_i)+F_1sin(x_5)-F_2cos(x_5)} \\
	    	-\frac{b}{I}F_1+\frac{a}{I}F_2
	    	\end{bmatrix}
            \label{car}
	    	\end{align}
   		   Our search is over the initial values in $[0,1]\times [0.5,1]$, and the input signals $F_1(t),F_2(t)\in [-1,1]$; other states are zero initially. Since the search over all the continuous input signals is a search in infinite dimension, here, we used piecewise constant inputs with 11 variables for each $F_1(t)$ and $F_2(t)$. So the overall search is over 24 dimensions.
  \ifthenelse{\boolean{Present}}{
  We start our search from the trajectory with $x_0=(0.5,0.6,0,0,0,0)$, and input signals $F_1(t)=0.2$, and $F_2(t)=0.1$ for $t\leq 7.2$ and $F_2(t)=-0.2$ for $t> 7.2$. This trajectory satisfies (\ref{req}) with the robustness value equal to 0.2950.}
  {Starting from a trajectory that satisfies Eq. (\ref{req}) with the robustness value equal to 0.2950, our method improves the robustness value to 0.8599 (Note that while in a falsification problem we try to decrease the robustness value, in a related problem called satisfaction problem increasing that value is desired).}
   \ifthenelse{\boolean{Present}}{
   Using our method with step size $h=0.02$, in the 8th iteration, the initial condition $x_0=(0.3983,0.6948,0,0,0,0)$ and the inputs shown in Fig (\ref{fig:caru}) are chosen and the robustness is improved to the value 0.8599 (Note that while in a falsification problem we try to decrease the robustness value, in a related problem called satisfaction problem increasing the robustness value is desired).
  }
   The projection of the trajectories into the $x_1-x_2$ plane is shown in Fig. \ref{fig:car}, where dark gray trajectories are refined to light gray ones. In Fig. \ref{fig:guardcar}, one can see that even if the trajectory from which we want to descend does not enter the goal set location, we are still able to improve the trajectory by descending toward the adjacent guard with the least distance from that set. 
           
         \ifthenelse{\boolean{Present}}{  \begin{figure}[t]
            \centering
            \begin{minipage}[b]{0.35\textwidth}
          \includegraphics[width=\linewidth]{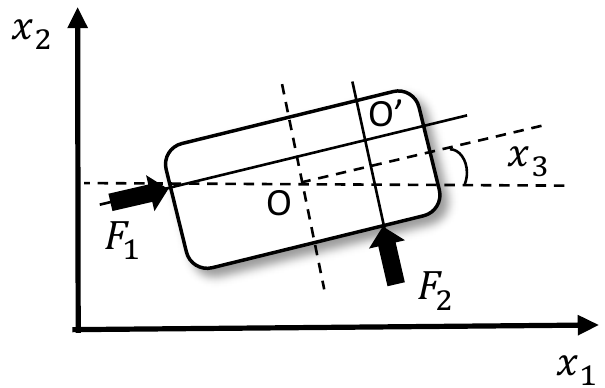}
				    	\caption{The rigid body with off centered thrusters. O is the center of the mass but the force is applied to O'.}
				    	\label{fig:thrust}
                        \vspace{10pt}
            \end{minipage}
            \hfill
            \begin{minipage}[b]{0.56\textwidth}
          \includegraphics[width=\linewidth]{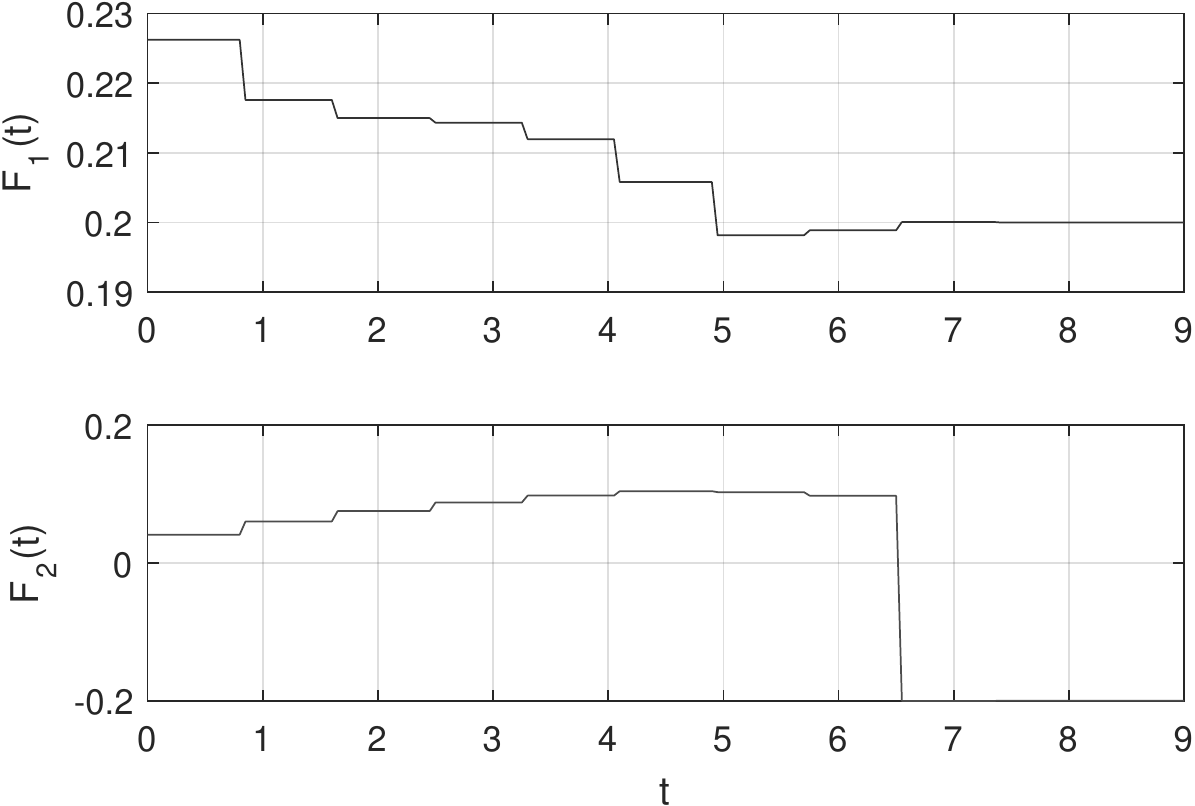}
	    	\caption{Final inputs in the process of improving the robustness valuation of the system of Eq. (\ref{car}) with respect to the specification $\varphi_2$.}
	    	\label{fig:caru}
            \end{minipage}
          \end{figure} }

 \begin{figure}[t]
            \centering
            \begin{minipage}[b]{0.47\textwidth}
	    	\includegraphics[width=1.05\linewidth]{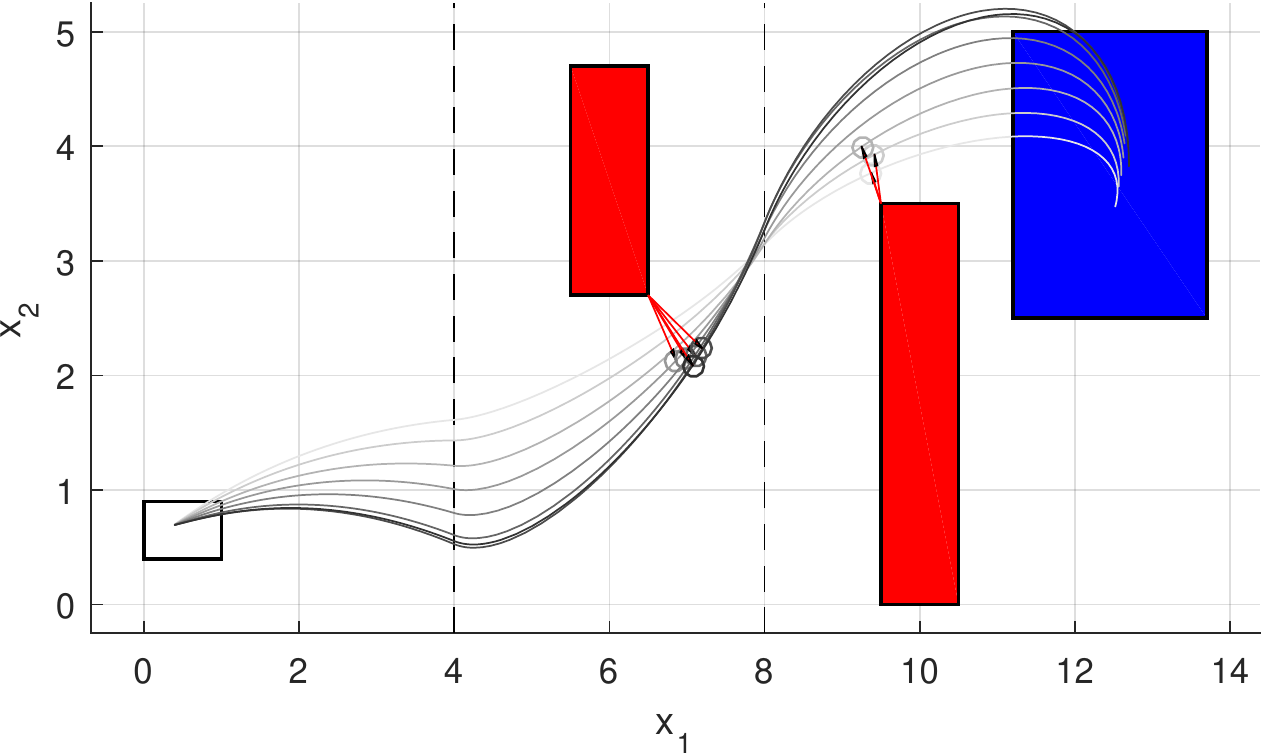}
	    	\caption{Improving the robustness of the trajectories of the system of Eq. (\ref{car}) with respect to the specification $\varphi_2$ from 0.2950 to 0.8599. Red arrows show the steepest ascent direction. }
	    	\label{fig:car}
            \end{minipage}
            \hfill
            \begin{minipage}[b]{0.49\textwidth}
          \includegraphics[width=1\linewidth]{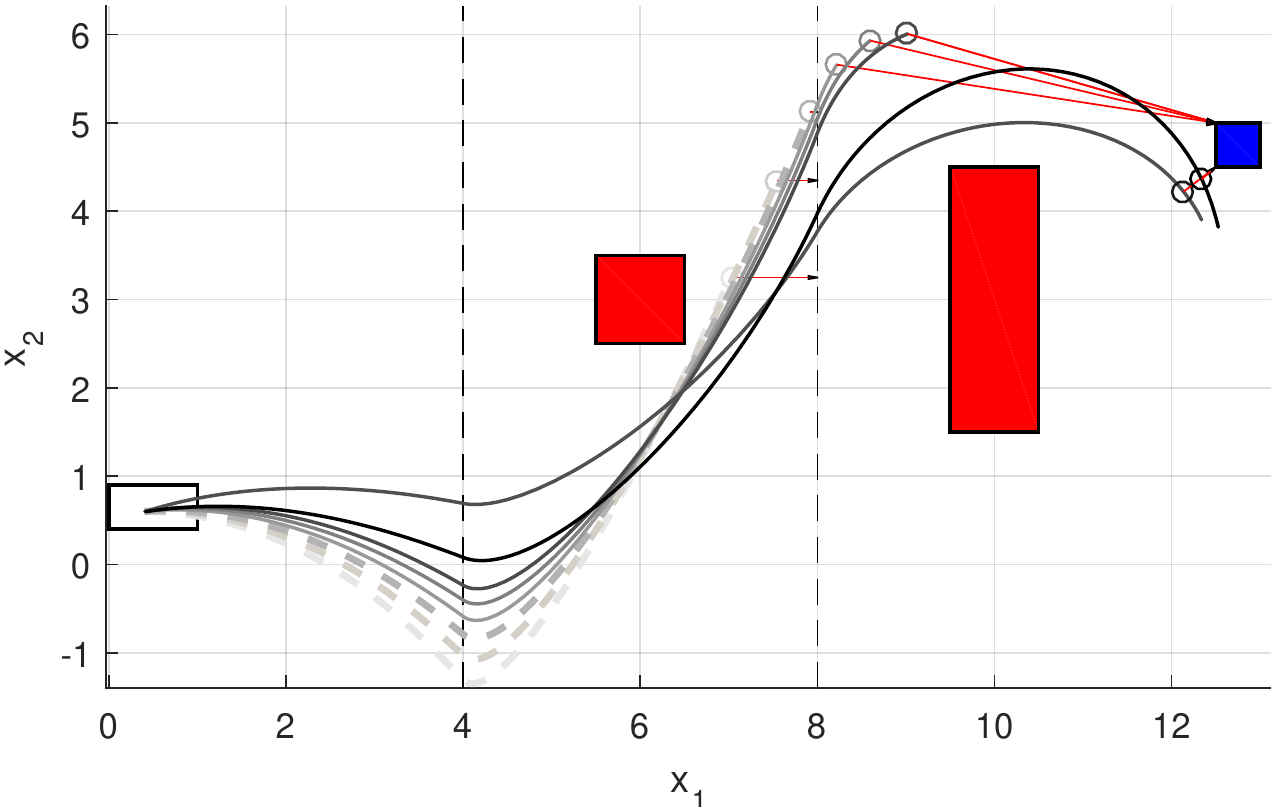}
	    	\caption{Trajectories that do not enter the goal set location (dashed trajectories here) can still improve by descending toward the guard set (dashed line at $x_1 = 8$).}
	    	\label{fig:guardcar}
            \end{minipage}
\end{figure}

	\end{example} 

In order to determine the effect of applying GD local search method to global search methods like Simulated Annealing (SA), we performed a statistical study in which we compare the combination of SA and GD (SA+GD) with SA only. To combine SA and GD, we apply GD algorithm whenever the samples taken by SA return a robustness value less than some threshold value $r_T$.

In our experiment we ran SA and SA+GD for 150 times with equal total number of samples $N=100$ and $r_T=2.5$\ifthenelse{\boolean{Present}}{. We set the parameters in Alg. ({\ref{alg:GD}}) to  $k_1=10$, $k_2=2$, $h=0.02$}{ }
 to automatically search for initial conditions and inputs that satisfy the specification $\varphi_2$ with $\Uc_1=[5.5,6.5]\times[2.5,3.5]$, $\Uc_2=[9.5,10.5]\times[1.5,4.5]$, $G=[12.5,13]\times[4.5,5]$ for the system in Example (\ref{ex3}).
In order to satisfy $\varphi_2$, we try to falsify its negation $\neg \varphi_2$. The results are shown in Table (\ref{tb:1}). The improvement in finding falsifying trajectories is clear from the total number of falsifications in the first row. Also, since GD gets a chance to improve the performance only if SA finds a robustness value less than $r_T$, we added the second row which shows in how many percents of the cases falsification is achieved if SA finds a robustness value less than $r_T$. While average of the best robustness value for all the tests is better for SA+GD algorithm, it is slightly better for SA if we only consider non-falsified cases. We can conclude that even if SA finds  small robustness values, it is hardly able to further decrease it. As the constant budget in the comparison is ``equal total number of simulations", we can claim that SA+GD can help improve the results if simulations/experiments are costly.  Choosing different design parameters might lead to even better experimental results.
 {\small\begin{table}[t]
	\renewcommand{\arraystretch}{1.3}
	\caption{Comparing SA and SA+GD results for the system of Example \ref{ex3}}
	\label{tb:1}
	\centering
    \scalebox{0.9}{
	\begin{tabular}{|c|c|c|}
		\hline
		Optim. method & SA & SA+GD \\
		\hline
	   num. of total falsification &  4/150   &  16/150 \\
       \hline
       \%  of falsification if SA finds $r\leq r_T$ & 13.33\%  & 39.02\%\\
		\hline
       Avg. min-Rob. (all the cases) &  9.1828   &  8.4818\\
		\hline
               Avg. min-Rob. (not falsified cases) &  9.4278  &  9.4968\\
		\hline
       min. min-Rob. (not falsified cases) &  0.0080  &  0.0059\\
		\hline
       max. min-Rob. (not falsified cases) &   13.1424  &   13.0880\\
		\hline
	\end{tabular}}
\end{table}}
\vspace{-5pt}
\section{Related Work}
One possible categorization for falsification approaches divides them into Single Shooting (SS) vs. Multiple Shooting (MS) methods. 
The technique of numerically solving boundary value problems is called shooting. SS approaches search over the space of system trajectories initiated from the set of initial conditions and under possible inputs. S-TaLiRo \cite{annpureddy2011s} and Breach \cite{donze2010breach} lie in this category. In contrast, MS approaches create approximate trajectories from trajectory segments starting from multiple initial conditions (not necessarily inside the initial set). 
Hence, the trajectories contain gaps between segments. 
The works \cite{zutshi2014multiple,zutshi2013trajectory} fall in in this category.
MS techniques cannot handle general TL requirements.

Motion planning approaches such as Rapidly-exploring Random Trees (RRT) lie in a category between SS and MS approaches.
Starting from an initial condition, the tree grows toward the unsafe set (or vice versa) to find an unsafe behavior of a non-autonomous system \cite{DreossiDDKJD15nfm,PlakuKV09tacas}. 
The applicability of these methods, however, is limited since it depends on many factors such as the dimensionality of the system, the modeling language, and the local planner.

\ifthenelse{\boolean{Present}}{Another possible categorization of falsification methods divides them into the following categories:
 \begin{enumerate}
 \item Methods that rely on optimizing a metric (called robustness) to systematically search for falsification: These methods try to minimize a robustness value which is assigned to each trajectory using global optimization techniques like SA and Cross Entropy. S-TaLiRo \cite{annpureddy2011s} and Breach \cite{donze2010breach} are among the tools that use this strategy. 
\item Methods that use constraint solvers to find falsification by translating the problem into constraint solving using Bounded Model Checking (BMC) approaches \cite{clarke1999model}. However, these approaches discretize the continuous dynamics and the resulting constraints become nonlinear even for linear hybrid systems by involving higher order terms, and,
\item Motion planning approaches such as Rapidly-exploring Random Trees (RRT): These methods lie in SS approach category where starting from an initial condition the tree grows toward the unsafe set (or vice versa) to find an unsafe behavior of a non-autonomous system \cite{DreossiDDKJD15nfm}. The applicability of these methods, however, is limited since it  depends on many factors such as system behavior itself and the used planner.
\vspace{-5pt}
 \end{enumerate}}

The performance of SS falsification methods can be improved using different complementary directions.
One direction is to provide alternative TL robustness metrics \cite{AkazakiH15cav}.
Another direction is to compute guaranteed or approximate descent directions \cite{YaghoubiF2017acc,abbas2013computing} in order to utilize descent optimization methods.
Our method in this paper is a SS approach that uses optimization and robustness metric to solve the falsification problem.
In \cite{abbas2013computing,abbas2014functional} robustness-based falsification is guided using descent direction; however, that line of work is only applicable to purely continuous systems. 
In \cite{abbas2011linear}, descent direction is calculated in the case of linear hybrid systems using optimization methods. 

In \cite{zutshi2013trajectory} authors use a MS approach to find falsifying trajectories of a hybrid system. Providing the gradient information to an NLP solver, they try to reduce the gaps between segments. Like our approach, they require knowledge of the system dynamics and solve a local search problem.
Unlike our method, in their approach, falsifying trajectories are segmented trajectories which are not real system trajectories unless the gaps between segments become zero in the optimization procedure (for systems with identity reset maps), which may not be the case, in general. 
As a result, falsification cannot be concluded unless they can randomly find a neighboring real system trajectory that violates the specification. We think that our approach can help their method to effectively search over real trajectories neighboring the segmented trajectory.   
Furthermore, the specifications they have focused on in \cite{zutshi2013trajectory} are safety properties and because of the nature of the search, their method cannot easily be extended to search for system trajectories that falsify general MTL formulas.

The general idea of using sensitivity to explore the parameter space of a problem that deals with robustness of a TL formula was first introduced in \cite{donze2010robust}. To solve a verification problem, they propose using the sensitivity of a robustness function to a parameter assuming that the function is differentiable to that parameter. 
There are however multiple factors which result in non-differentiability of the robustness function with respect to a parameter: First of all, the predicates themselves might be non smooth and non-differentiable. Secondly, hybrid systems may have non smooth and non-differentiable trajectories. Finally, logical operators in the TL formula impose $min$ and $max$ operators to robustness function. The paper suggests using left and right hand derivatives for dealing with $min$ and $max$ operators, but it does not propose solutions for the first two cases. 
In our framework, by introducing Eq. (\ref{jcon}), we solve the non differentiability issue in the first case and the analysis in Sec. \ref{sec4} deals with this issue in the second case. 
Also, the problem we try to solve is a different problem (a falsification problem).  

In \cite{pant2017smooth}, a smooth infinitely differentiable robustness function is introduced which solves -- to some extent -- the non-differentiability problem of the robustness function to parameters. In the case of hybrid systems however, we still deal with this problem as the non-differentiability is caused by the system model rather than the robustness function itself. 
In the future, we will investigate if the results in \cite{pant2017smooth} could further improve the performance of gradient descent falsification methods as formulated in our work.

In \cite{donze2009parameter}, an algorithm to approximate reachable sets using sensitivity analysis is introduced. 
Sensitivity of hybrid systems without reset maps is used to verify safety properties. 
Like all approaches that try to solve a coverage problem, the method suffers from the state explosion issue which happens when one tries to cover the high dimensional spaces induced by the variables in the input signal parameterization.
Our framework solves a different problem and it is applicable to hybrid systems with reset maps under general TL formulas. 
Furthermore, as we are not solving a coverage problem, we do not face the state explosion issue.

\section{Conclusion}
 \vspace{-5pt}
TL robustness guided falsification \cite{AbbasFSIG13tecs} has shown great potential in terms of black or gray box automatic test case generation for CPS \cite{FainekosSUY12acc,StrathmannO15arch,SankaranarayananEtAl2016medcps}.
In this paper, we presented a method that locally improves the search for falsifying behaviors by computing descent directions for the TL robustness in the search space of the falsification problem.
Our proposed method computes such descent directions for non-linear hybrid systems with external inputs, which was not possible before in the literature.
Using examples, we demonstrated that our framework locally decreases the TL robustness at each iteration.
Furthermore, our preliminary statistical results indicate that it is possible to improve a global test-based falsification framework when the proposed local gradient descent method is utilized.

Currently, the proposed framework requires a symbolic representation of the non-linear dynamics and the switching conditions of the hybrid automaton in order to compute the descent direction. 
As future research, we expect that we can relax this requirement by numerically computing approximations to the descent directions similarly to our work for smooth non-linear dynamical systems \cite{YaghoubiF2017acc}.
This will enable the application of the local descent method to a wide range of Simulink models without explicit extraction of the system dynamics.


\vspace{-5pt}
\section*{Acknowledgments}
\vspace{-5pt}
This work was partially supported by the NSF awards CNS-1319560, CNS 1350420, IIP-1361926, and the NSF I/UCRC Center for Embedded Systems.
%
%
\bibliographystyle{splncs}
\bibliography{hybref_1}

\end{document}